\documentclass{article}
\usepackage{graphicx}
\usepackage{hyperref}
\usepackage[utf8]{inputenc}
\usepackage[svgnames]{xcolor}
\usepackage{amsthm}
\usepackage{amssymb}
\usepackage{array}
\usepackage{moresize}

\usepackage{subcaption}

\newcommand{\define}[1]{\emph{#1}}
\newcommand{\mA}{\ensuremath{\mathcal{A}}}
\newcommand{\mB}{\ensuremath{\mathcal{B}}}
\newcommand{\mG}{\ensuremath{\mathcal{G}}}
\newcommand{\mH}{\ensuremath{\mathcal{H}}}
\newcommand{\mI}{\ensuremath{\mathcal{I}}}
\newcommand{\mP}{\ensuremath{\mathcal{P}}}

\newcommand{\net}{\emph{net}}
\newcommand{\codomino}{\emph{co-domino}}
\newcommand{\domino}{\emph{domino}}
\newcommand{\wheel}{\emph{wheel}}

\newtheorem{observation}{Observation}
\newtheorem{conjecture}{Conjecture}
\newtheorem{theorem}{Theorem}
\newtheorem{lemma}{Lemma}
\newtheorem{corollary}{Corollary}
\def\titulo{Biclique Graphs of $K_3$-free Graphs and Bipartite Graphs}
\begin{document}
\thispagestyle{empty}
\begin{center}
\Large
 {\bf \titulo\footnote{This work was partially supported by
  ANPCyT (PICT-2013-2205), % MG
  CONICET, % MG
  CAPES and % MG
  CNPq (428941/2016-8). % Universal
}
}
\vspace*{1cm}

\large
Marina Groshaus \\
 \normalsize
Universidade Tecnológica Federal do Paraná, Brazil\\
\texttt{marinagroshaus@yahoo.es} \\

\vspace*{1cm}

\large
A. L. P. Guedes \\ % \orcidID{0000-0001-5449-5393}
 \normalsize
Universidade Federal do Paraná, Brazil\\
\texttt{andre@inf.ufpr.br}\\
\vspace*{.5cm}

\end{center}

\begin{abstract}
  A \define{biclique} of a graph is a maximal complete bipartite subgraph.
  The  \define{biclique graph}  of a  graph $G$,  $KB(G)$, defined  as the
  intersection  graph  of  the  bicliques of  $G$,  was  introduced  and
  characterized in 2010.
  However,  this  characterization  does  not lead  to  polynomial  time
  recognition  algorithms.   The  time  complexity  of  its  recognition
  problem  remains open.   There are  some  works on  this problem  when
  restricted to some classes.
  In this  work we give  a characterization of  the biclique graph  of a
  $K_3$-free graph $G$.  We prove that  $KB(G)$ is the square graph of a
  particular graph which we  call \define{Mutually Included Biclique Graph}
  of $G$  ($KB_m(G)$).  Although it does  not lead to a  polynomial time
  recognition  algorithm, it  gives a  new tool  to prove  properties of
  biclique  graphs   (restricted  to  $K_3$-free  graphs)   using  known
  properties of square graphs. 
  For instance  we generalize a property  about induced ${P_3}'$s in biclique
  graphs to a property about stars and proved a conjecture posted 
  by Groshaus and Montero, when restricted to $K_3$-free graphs.
  Also we characterize the class of biclique graphs of bipartite
  graphs.  We prove  that  $KB($bipartite$) =  ($IIC-comparability$)^2$,
  where \define{IIC-compa\-rability}  is a subclass of  comparability graphs
  that we call \define{Interval Intersection Closed Comparability}.
\end{abstract}

\normalsize
\vspace*{1cm}
\newpage

  % \MSC[2010]
  %      05C75 % Structural characterization of families of graphs
  % \sep 05C76 % Graph operations (line graphs, products, etc.)
  % \sep 05C85 % Graph algorithms (graph theory)
  % \sep 68R10 % Graph theory in connection with computer science (including graph drawing)
  % \sep 68Q25 % Analysis of algorithms and problem complexity

\section{Introduction}

A \define{biclique}  of a graph is  a vertex set that  induces a maximal
complete  bipartite subgraph.   The \define{biclique  graph} of  a graph
$G$, denoted by  $KB(G)$, is the intersection graph of  the bicliques of
$G$.
The  biclique   graph  was   introduced  by  Groshaus   and  Szwarcfiter
\cite{Groshaus2010}, based on the concept of clique graphs.  They gave a
characterization of biclique graphs  (in general) and a characterization
of biclique graphs of bipartite graphs.
The  time  complexity of  the  problem  of recognizing  biclique  graphs
remains open.

%------------------------------------------------------------
% copiado do artigo para a DAM para bi-intervalos
%------------------------------------------------------------
Bicliques in  graphs have applications  in various fields,  for example,
bio\-lo\-gy: protein-protein  interaction networks~\cite{Bu2003}, social
networks:        web        community        discovery~\cite{Kumar1999},
genetics~\cite{Atluri2010},  medicine~\cite{Nagarajan2008},  information
theory~\cite{Haemers2001}.  More applications  (including some of these)
can be found in the work of Liu, Sim and Li~\cite{Liu2006}.

The biclique graph can be considered  as a graph operator: given a graph
$G$,  the operator  $KB$  returns  the biclique  graph  of $G$,  $KB(G)$
\cite{GroshausGuedesMontero2016,Groshaus2013,Groshaus2009}.
Some problems related to graph operators are studied in relation to some
classes of graphs.   Given a graph operator $\mH$ and  a class $\mA$, it
is  studied the  problem of  recognizing  the class  $\mH(\mA)$, or  the
problem of recognizing the class $\mH^{-1}(\mA)$.
These problems have been studied in the context of the clique graph, $K$
(clique  operator).  There  are works  about $K(\mA)$,  where $\mA$  is
clique-Helly,  chordal,  interval,   split,  diamond-free,  dismantable
graphs,               arc-circular              graphs               etc
\cite{Bondy2003,Duran2001,Larrion2002,Lin2010,Prisner1999}.

There are  few works  in the literature  about recognizing  the biclique
graphs            of            some            graph            classes
\cite{CruzGroshausGuedesPuppo2020,GroshausGuedesPuppo2016,Groshaus2006,Puppo2019}.
%------------------------------------------------------------

In this  work we prove that  every biclique graph of  a $K_3$-free graph
(triangle-free)   is    the   square    of   some   graph,    that   is,
$KB(K_3$-free$) \subset  (\mG)^2$, where $\mG$  denote the class  of all
graphs and $(\mA)^2$ denote the class of the square graphs of the graphs
of the class $\mA$.  This result gives a tool for studying other classes
of biclique graphs.
% Sugestão da Marina
Known  results on  square  graphs follow  directly  for biclique  graphs
(using  known  properties of  square  graphs).   Also, some  results  on
biclique graphs of graphs, when  restricted to $K_3$-free graphs, can be
easily proven using the fact that it is a square graph. For example, the
fact  that  every  $P_3$  is  contained   in  a  diamond  or  a  $3$-fan
\cite{Groshaus2010}, the fact that the  number of vertices of degree $2$
is  at most  $n/2$ and  the  fact that  the family  of neighbourhoods  of
vertices     of    degree     2    satisfy     the    Helly     property
\cite{GroshausMontero2019}. 

Groshaus and Montero presented a conjecture stating that
a certain structure is forbidden in biclique graphs \cite[Conjecture 6.2]{GroshausMontero2019}.
We prove that this conjecture  holds when  $G$ is a  $K_3$-free graph,
using the fact that $H$ is the square of some particular  graph.

\begin{conjecture}[\cite{GroshausMontero2019}]
  \label{conj.montero}
  If $H=KB(G)$  for some graph $G$,  where $H$ is not  isomorphic to the
  diamond then there do not exist  $v_1, v_2, \ldots, v_n \in V(H)$ such
  that $N_H(v_1)  = N_H(v_2) =  \cdots = N_H(v_n)$ and  their neighbours
  are contained in a $K_n$ for $n \geq 2$.
\end{conjecture}

A  \define{comparability graph}  $G$ is  such that  there is  a partially
ordered set (poset) $(V(G),\leq_G)$ where $uv \in E(G)$ if and
only if $u$ and $v$ are comparable by $\leq_G$ \cite{Brandstaedt1999}.
%
% A  graph $G$  is a  \define{co-comparability graph}  if its  complement,
% $\overline{G}$, is a comparability graph.
%
We define a subclass of comparability graphs, the class
of      \define{interval     intersection      closed     comparability}
(IIC-comparability)  graphs, and  we prove  that the  class of  biclique
graphs of bipartite graphs is equal to the class of the square graphs of
IIC-comparability              graphs,             that              is,
$KB($bipartite$)    =     ($IIC-comparability$)^2$,    giving    another
characterization of biclique graphs of bipartite graphs.

\subsection{Some Definitions and Notations}

Let $G$ be a  graph with vertex set $V(G)$ and  edge set $E(G)$.  Denote
the set of \define{neighbours} of a vertex $v \in V(G)$ as $N_G(v)$.
Let $S \subseteq  V(G)$ and define $N_G^*(S)$ to be  the set of vertices
of $G$ that are incident to every vertex of $S$. That is,
$N_G^*(S) = \bigcap_{v \in S} N_G(v).$
Note that if $U \subseteq S$, $N_G^*(S) \subseteq N_G^*(U)$.

Given   a    poset   $\mP   =    (C,\leq)$   and   $x   \in    C$,   let
$I_{\mP}^-(x)    =    \{y    \in    C   \mid    y    \leq    x\}$    and
$I_{\mP}^+(x)  = \{y  \in  C  \mid x  \leq  y\}$  be, respectively,  the
\define{predecessors interval} and \define{successors interval} of $x$ in
$\mP$.
We say  that a poset  $\mP = (C,\leq)$ is  \define{interval intersection
  closed            (IIC)}             if            the            sets
$\mI_{\mP}^-    =     \{I_{\mP}^-(x)    \mid    x    \in     C\}$    and
$\mI_{\mP}^+ =  \{I_{\mP}^+(x) \mid  x \in  C\}$ are  \define{closed under
  intersection}. That is,  for every pair $u$, $v \in  C$, the following
sentences are true:
\begin{itemize}
\item if $I_{\mP}^-(u) \cap I_{\mP}^-(v) \neq \emptyset$ then there is a
  $w \in C$ such that $I_{\mP}^-(w) = I_{\mP}^-(u) \cap I_{\mP}^-(v)$);
\item if $I_{\mP}^+(u) \cap I_{\mP}^+(v) \neq \emptyset$ then there is a
  $w \in C$ such that $I_{\mP}^+(w) = I_{\mP}^+(u) \cap I_{\mP}^+(v)$).
\end{itemize}

%%%% VERSÃO COM MÁXIMO e MÍNIMO
%%%%%%%%%%%%%%%%%%%%%%%%%%%%%%%%%%%%%%%%%%%%%%%%%%%%%%%%%%%%%%%%%%%%%%%%%%
% Note that it is  equivalent to say that, for every pair  $u$, $v \in C$,
% the following sentences are true:
% \begin{enumerate}
% \item   if  $I_{\mP}^-(u)   \cap  I_{\mP}^-(v)   \neq  \emptyset$   then
%   $I_{\mP}^-(u) \cap I_{\mP}^-(v)$ has a maximum;
% \item   if  $I_{\mP}^+(u)   \cap  I_{\mP}^+(v)   \neq  \emptyset$   then
%   $I_{\mP}^+(u) \cap I_{\mP}^+(v)$ has a minimum.
% \end{enumerate}
%%%%%%%%%%%%%%%%%%%%%%%%%%%%%%%%%%%%%%%%%%%%%%%%%%%%%%%%%%%%%%%%%%%%%%%%%%

Let  the graph  class  \define{IIC-comparability (Interval  Intersection
  Closed  Comparability)}  be the  class  of  comparability graphs  with
posets that are IIC.

\section{Mutually Included Biclique Graph}

Denote a  biclique $P$ of $G$,  with bipartition $(X,Y)$, as  $XY$. That
is, $XY = X \cup Y$.
Given  a  biclique   $P=XY$  and  a  vertex  $v  \not\in   P$  then  (i)
$v \not\in N_G^*(X) \cup N_G^*(Y)$ or  (ii) there are vertices $x \in X$
and $y \in  Y$ such that $x, y  \in N_G(v)$. Note that in  case (ii) the
vertices $x$, $y$ and $v$ form a $K_3$. So, if $G$ is a $K_3$-free graph
(ii) is always false and if $v \not\in P$ then (i) holds.

% \comment{REV1}{green}{The  proofs   of  Lemma   1  and  corollary   1  are
%   straightforward, and  I think  maybe the text  would profit  more from
%   additional figures, or explaining better some other parts; \\
%   - Tirei prova do lema e do corolário. \\
%   - Transformar lema em observação?}
%\ag{Transformar em observação (sem a prova)?}
\begin{observation}\label{obs.Nstar}%\margem{Poderia ser uma observação sem prova?}
  Given a  $K_3$-free graph $G$, two  independent sets, $X$ and  $Y$, of
  $G$ form  a biclique $XY$  of $G$  iff $N_G^*(X) =  Y$ and
  $N_G^*(Y) = X$.
\end{observation}
%\begin{proof}
%   % Let $G$ be a  $K_3$-free graph and let $X$ and  $Y$ be two independent
%   % sets of $G$.
% %
%\ag{Tirar prova se precisar de espaço.}
%  Suppose $XY$ is  a biclique of $G$.   There is no vertex  $u \in V(G)$
%  such that  $u \in N_G^*(X) \setminus  Y$, otherwise $XY$ would  not be
%  maximal.   Also,   there  is  no   vertex  $v  \in  V(G)$   such  that
%  $v \in Y  \setminus N_G^*(X)$, otherwise $XY$ would not  be a complete
%  bipartite subgraph. So $N_G^*(X) = Y$.
% %
%  Changing     roles    of     $X$     and    $Y$     we    get     that
%  $N_G^*(Y) = X$.
% %
%  Now, suppose  $N_G^*(X) = Y$ and  $N_G^*(Y) = X$. By  definition, $XY$
%  induces  a complete  bipartite  subgraph.  Suppose  that  $XY$ is  not
%  maximal.   Then   there  is  a   vertex  $w  \not\in  XY$   such  that
%  $XY  \cup \{w\}$  also induces  a complete  bipartite subgraph.   Then
%  $w  \in  N_G^*(X)   \cup  N_G^*(Y)$  and  $w  \in  XY$,   which  is  a
%  contradiction.  Consequently, $XY$  is maximal  and is  a biclique  of
%  $G$. \qed
%\end{proof}

Note  that  no part  of  a  biclique  intersects both parts  of  another
biclique. That is, if $P = X_PY_P$ and $Q = X_QY_Q$ are two bicliques of
a graph $G$, $X_P \cap X_Q = \emptyset$ or $X_P \cap Y_Q = \emptyset$.
So,    assume     that    $X_P    \cap    Y_Q     =    \emptyset$    and
$X_Q  \cap   Y_P  =   \emptyset$.   We   say  that   $P$  and   $Q$  are
\define{mutually  included}  if  $X_Q  \subset  X_P$  and  $Y_P  \subset
Y_Q$. See Figure~\ref{fig:mutuallyAB}.

\begin{figure}[htb]
\centering
\begin{subfigure}{.5\textwidth}
  \centering
  \includegraphics[width=4cm]{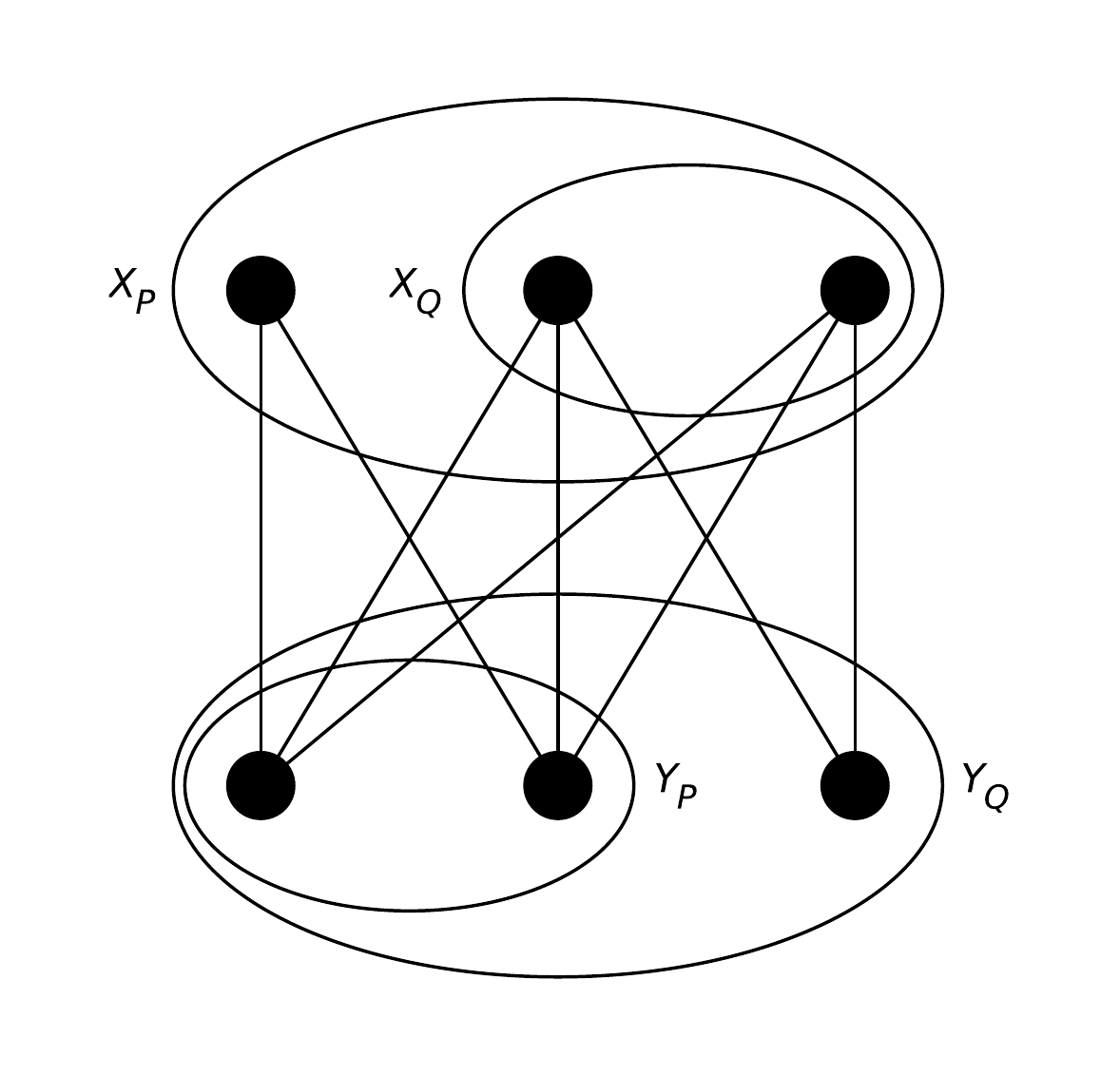}
  \caption{}
  \label{fig:mutuallyAB}
\end{subfigure}%
\begin{subfigure}{.5\textwidth}
  \centering
  \includegraphics[width=5cm]{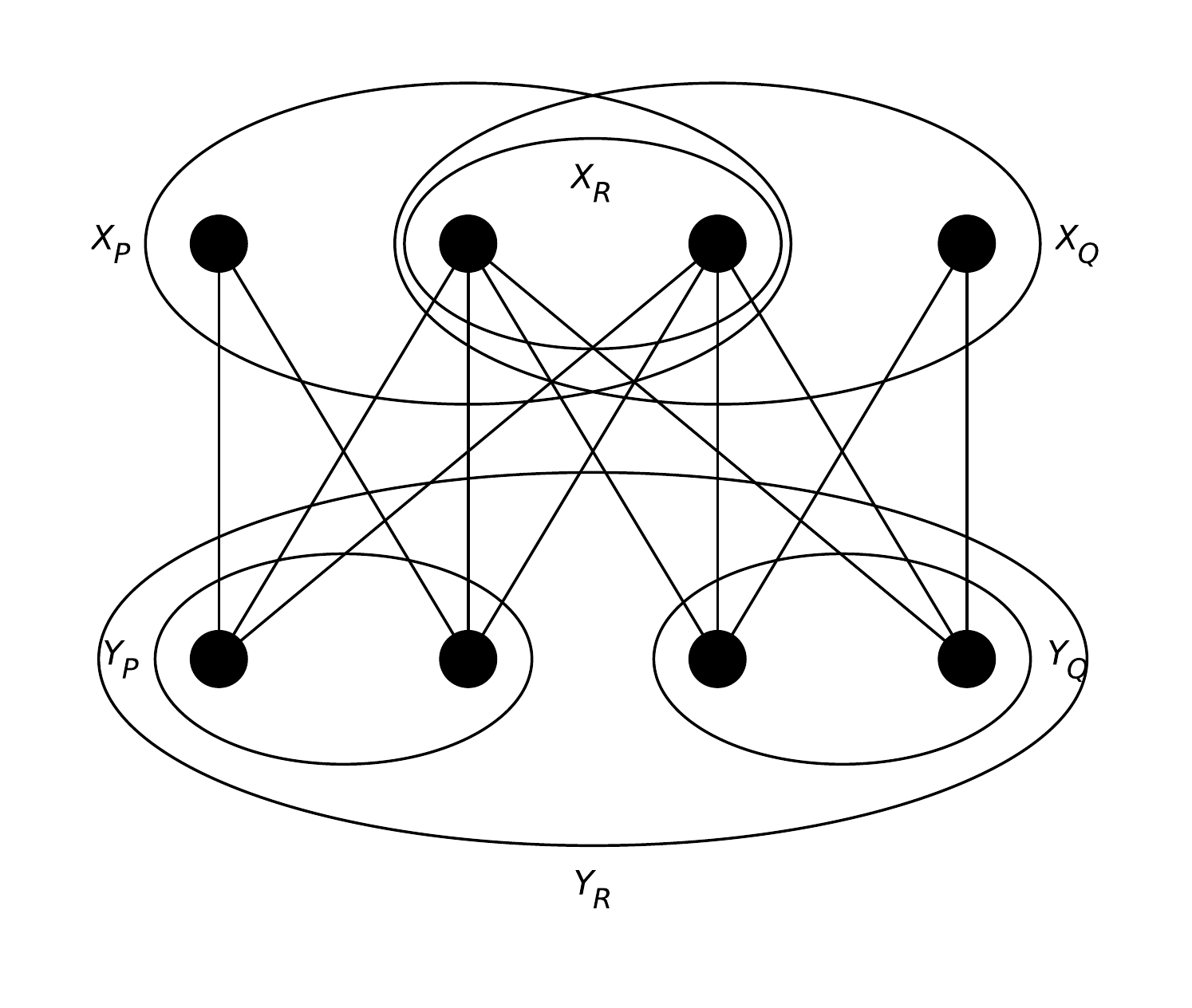}
  \caption{}
  \label{fig:not-mutuallyAll}
\end{subfigure}%
\caption{(a) Mutually included bicliques (b) Two   intersecting    not   mutually
  included bicliques, $P$  and $Q$ and a biclique  $R$ mutually included
  with both.}
\end{figure}

Define $KB_m(G)$, the \define{mutually  included biclique graph of $G$},
as the graph with  the bicliques of $G$ as its  vertex set and $PQ$
is an edge iff  $P$ and $Q$ are mutually included.  Note that
$KB_m(G) \subseteq KB(G)$ (with the same vertex set).

% \comment{REV1}{green}{Proof of  Lemma 2,  second paragraph: why  can’t the
%   chosen vertex x be in $X_P$? \\
% - $x$  was choosen  such that  $x \in N_G^*(Y_R)  \setminus X_R$.  So $x
% \not\in X_R$ and consequently also not in $X_P$.}
\begin{lemma}\label{lem.bicliqueintersecting}
  If $P  = X_PY_P$ and  $Q = X_QY_Q$ are  two bicliques of  a $K_3$-free
  graph   $G$,    that   are   not   mutually    included,   such   that
  $X_P \cap  X_Q \neq \emptyset$ then  there is a biclique  $R = X_RY_R$
  such that $X_R = X_P \cap X_Q$ and $Y_R \supseteq Y_P \cup Y_Q$.
\end{lemma}
\begin{proof}
  Let $G$ be a $K_3$-free graph and let $P = X_PY_P$ and $Q = X_QY_Q$ be
  two  bicliques  of  $G$  that  are  not  mutually  included  and  that
  $X_P  \cap   X_Q  \neq  \emptyset$.    Let  $X_R  =  X_P   \cap  X_Q$,
  $Y_R = N_G^*(X_R)$.  Note that $Y_R$  is an independent set (as $G$ is
  a  $K_3$-free graph)  and  that  $(Y_P \cup  Y_Q)  \subseteq Y_R$  (by
  definition of $N_G^*(X_R)$). See Figure~\ref{fig:not-mutuallyAll}.

  Suppose  $R   =  X_RY_R$  is  not   a  biclique  of  $G$.    Then,  by
  Observation~\ref{obs.Nstar},   $N_G^*(Y_R)   \neq  X_R$.    By   definition,
  $X_R    \subseteq    N_G^*(Y_R)$,    so     there    is    a    vertex
  $x \in N_G^*(Y_R) \setminus X_R$.  As $x$ is neighbour of every vertex
  of  $Y_R$,  it is  also  neighbour  of  every  vertex of  $Y_P$,  then
  $P \cup \{x\}$ induces a complete  bipartite subgraph and $P$ is not a
  biclique. So,  there is  no such  vertex $x$,  $N_G^*(Y_R) =  X_R$ and
  $R = X_RY_R$ is a biclique of $G$. \qed
\end{proof}

\begin{corollary}\label{cor.inter.mutuallyorR}
  If $P$  and $Q$ are two  intersecting bicliques of a  $K_3$-free graph
  $G$ then $P$ and $Q$ are mutually  included or there is a biclique $R$
  that is mutually included both with $P$ and $Q$.
\end{corollary}
% \begin{proof}
%   Let $G$ be a $K_3$-free graph and let $P = X_PY_P$ and $Q = X_QY_Q$ be
%   two intersecting bicliques.
% %
%   Suppose $P$ and $Q$ are  not mutually included.  Also suppose, without
%   loss  of  generality,   that  $X_P  \cap  X_Q   \neq  \emptyset$.   By
%   Lemma~\ref{lem.bicliqueintersecting} there is a  biclique $R = X_RY_R$
%   such  that $X_R  = X_P  \cap X_Q$  and $Y_R  \supseteq Y_P  \cup Y_Q$.
%   Moreover,  $R$ and  $P$, and  $R$ and  $Q$ are  two pairs  of mutually
%   included bicliques. \qed
% \end{proof}

% \comment{REV1}{green}{Proof of Theorem 1: explain better why when P and Q do not intersect, then there is no biclique mutually included with both;}
% \comment{REV3}{green}{Page 5. Proof of theorem 1: I would recommend more detail on why the first sentence of the second paragraph is correct.}
%% Included by onservations of 2 reviewers
\begin{lemma}\label{lem.commonMutually}
  If  $P$ and  $Q$  are two  bicliques  such that  there  is a  biclique
  mutually     included    with     both    $P$     and    $Q$,     then
  $P \cap Q \neq \emptyset$.
\end{lemma}
\begin{proof}
  Suppose $P = X_PY_P$ and $Q = X_QY_Q$.  Let $R = X_RY_R$ be a biclique
  mutually   include  with   $P$   and  $Q$.    Suppose  w.l.o.g.   that
  $X_P \subset X_R$  and $Y_R \subset Y_P$. Also suppose  that $X_Q \cap
  Y_R = \emptyset$. If $X_R \subset  X_Q$, then $X_P \subset X_Q$ and $P
  \cap Q \neq \emptyset$. On the  other hand, if $X_Q \subset X_R$, then
  $Y_R \subset Y_Q$ and $Y_R \subseteq Y_P  \cap Y_Q$ and also $P \cap Q
  \neq \emptyset$.
  \qed
\end{proof}

\begin{theorem}\label{thm.KBm-square}
  If $G$ is a $K_3$-free graph, then $KB(G) = (KB_m(G))^2$.
\end{theorem}
\begin{proof}
  Let $G$ be a $K_3$-free graph and let $P$ and $Q$ be
  two bicliques.

% \comment{REV3}{green}{Page  5. Proof  of theorem  1: "there  a biclique"
%   $\to$ "there is a biclique". \\
% - Done}
  If  $P$ and  $Q$ intersect,  by Corollary~\ref{cor.inter.mutuallyorR},
  $P$ and $Q$ are mutually included or there is a biclique $R$ that is
  mutually included with  both. That is, $P$ and $Q$  are at distance at
  most $2$ in $KB_m(G)$.

  If  $P$ and  $Q$ do  not intersect,  by Lemma~\ref{lem.commonMutually}
  there is no biclique that is mutually included with both. That is, $P$
  and $Q$ are at distance at least $3$ in $KB_m(G)$.

  So $KB(G) = (KB_m(G))^2$. \qed
\end{proof}

\begin{corollary}
  $KB(K_3$-free$) \subsetneq (\mG)^2$.
\end{corollary}
\begin{proof}
  By Theorem~\ref{thm.KBm-square}, $KB(K_3$-free$) \subseteq  (\mG)^2$.

  Observe that \net$^2  \in (\mG)^2$ but \net$^2  \not\in KB(\mG)$ (from
  the   work  of   Montero  \cite{Montero2008},   by  inspection).   See
  Figure~\ref{fig:netgraph}.
  % using the characterization presented by Groshaus and Szwarcfiter
  % \cite{Groshaus2010}
  % That is, $KB(K_3$-free$) \subsetneq (\mG)^2$.
  \qed
\end{proof}

\begin{figure}[htb]
\centering
\begin{subfigure}{.33\textwidth}
  \centering
  \includegraphics[width=3cm]{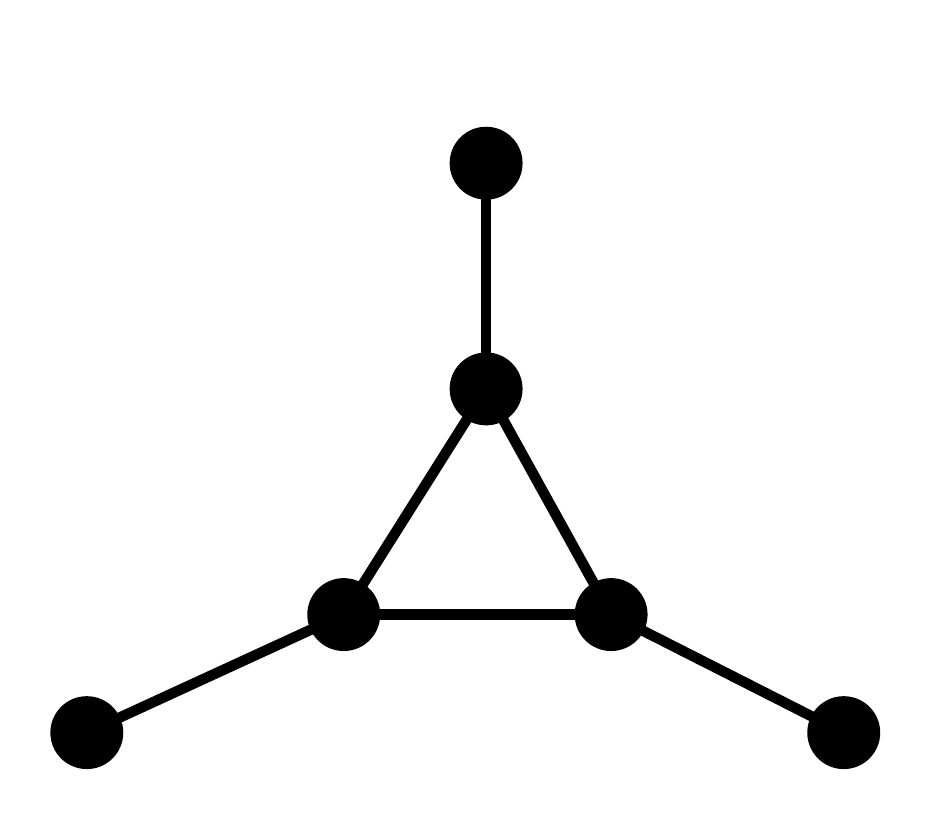}
  \caption{}
  \label{fig:netgraph}
\end{subfigure}%
\begin{subfigure}{.33\textwidth}
  \centering
  \includegraphics[width=3cm]{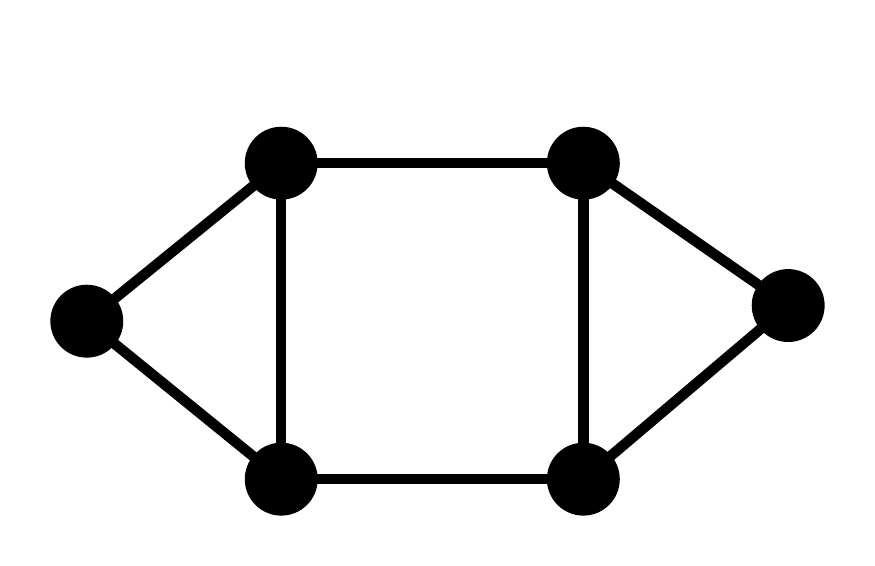}
  \caption{}
  \label{fig:co-domino}
\end{subfigure}%
\begin{subfigure}{.33\textwidth}
  \centering
  \includegraphics[width=2cm]{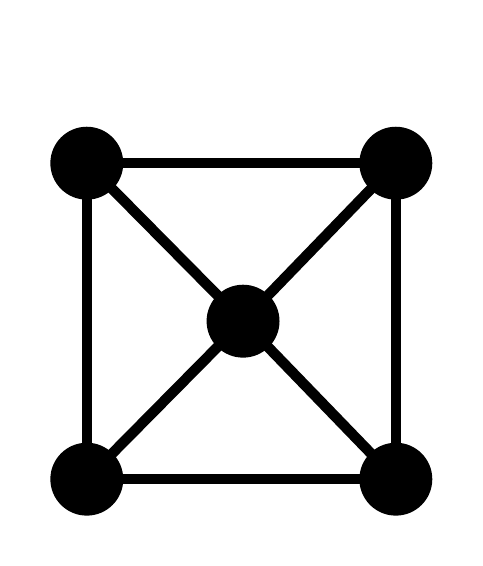}
  \caption{}
  \label{fig:4wheel}
\end{subfigure}
\caption{(a) \net\ graph (b) \codomino\ graph, the  complement of the
  \domino\ graph (c) $4$-\wheel\ graph}
\label{fig:test}
\end{figure}

In general, it is  not the case that the biclique graph  is the square of
some   graph.   For   instance,   consider  the   \codomino\  graph   of
Figure~\ref{fig:co-domino}.   The biclique  graph of  \codomino, is  the
$4$-\wheel\     graph     of    Figure~\ref{fig:4wheel},     that     is
$KB($\codomino$) = 4$-\wheel.  But $4$-\wheel\  is not the square of any
graph (by inspection).

\section{Properties}     % of ... Mutually Included Biclique Graph}

In this section we present properties of mutually included biclique graphs, 
square graphs, and therefore,
for biclique graphs of triangle-free graphs. 
%\ag{Bla bla} ...

\begin{lemma}\label{lem.mutually-subset}
  Let $P  = X_PY_P$ and  $Q = X_QY_Q$ be  two bicliques of  a $K_3$-free
  graph $G$ such  that $X_P \cap Y_Q  = Y_P \cap X_Q  = \emptyset$.  $P$
  and $Q$  are mutually  included iff  $X_P \subset  X_Q$ or
  $X_Q \subset X_P$.
\end{lemma}
\begin{proof}
  Let  $P  =  X_PY_P$ and  $Q  =  X_QY_Q$  be  two bicliques  such  that
  $X_P \cap Y_Q = Y_P \cap X_Q = \emptyset$.

  By   definition  if   $P$   and  $Q$   are   mutually  included   then
  $X_P \subset X_Q$ or $X_Q \subset X_P$.

  Now        suppose        $X_P       \subset        X_Q$.         Then
  $N_G^*(X_Q)   \subseteq   N_G^*(X_P)$,    by   the   definition.    By
  Observation~\ref{obs.Nstar},  $Y_P =  N_G^*(X_P)$,  $Y_Q  = N_G^*(X_Q)$  and
  $Y_P \neq Y_Q$.   So, $Y_Q \subset Y_P$. Consequently $P$  and $Q$ are
  mutually included.

  Changing roles of $P$  and $Q$, we conclude that $P$  and $Q$ are also
  mutually included in the case when $X_Q \subset X_P$.
  \qed
\end{proof}

%  \comment{REV1}{green}{Sometimes  it  is  hard  to  know  if  you  are
%  restricting yourselves  to  K3-free  or bipartite  graphs,  or if  no
% constraint  is considered.  For instance,  in the  paragraph preceding
% Lemma 4 you say that  the union of mutually included bicliques induces
% a bipartite graph. I guess you’re talking about K3-free, but it is not
% clear;}
Note that, for  any graph, the union of two  mutually included bicliques
induces a bipartite  graph and a set of mutually  included bicliques are
``nested''.
In  Figure~\ref{fig:mutuallynested} are  presented a  set of
``nested'' mutually included bicliques.

\begin{figure}[htb]
  \centering
  \includegraphics[width=4cm]{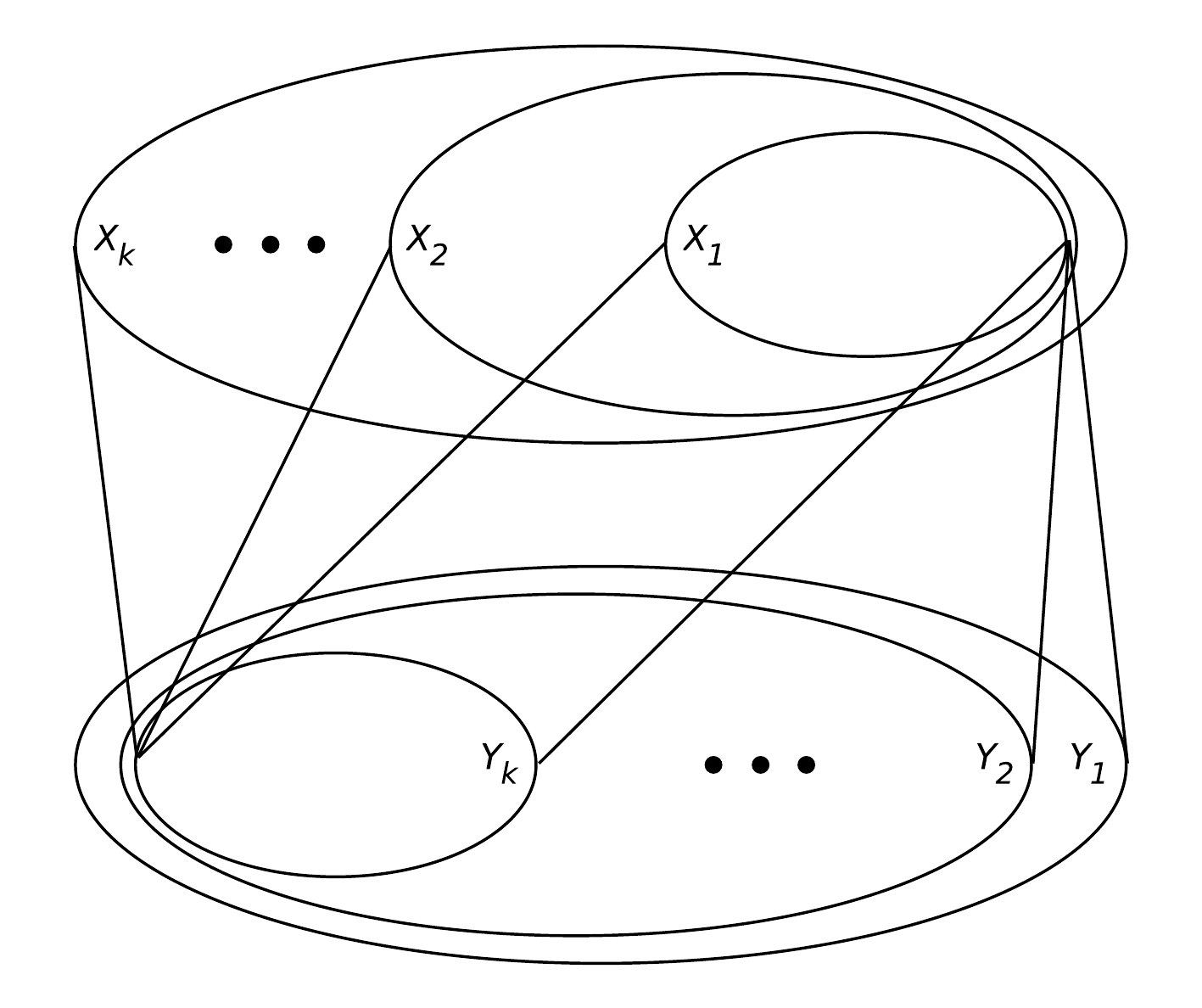} 
  \caption{Nested mutually included bicliques.}
  \label{fig:mutuallynested}
\end{figure}

\begin{lemma}\label{lem.nested}
  For any graph $G$,  and a set $C$ of bicliques of  $G$ such that every
  two of them  are mutually included, then there is  an ordering of $C$,
  $(P_1, \ldots, P_k)$  such that $X_1 \subset \cdots  \subset X_k$ (and
  $Y_k \subset  \cdots \subset Y_1$),  assuming that $P_i=  X_iY_i$, for
  $1 \leq i \leq k$.
\end{lemma}
\begin{proof}
  Let $G$ be  any graph and $C$ be  a set of bicliques of  $G$ such that
  every two  of them  are mutually included.   Recall that  the relation
  $\subset$ is  a partial order and  note that the bicliques  of $C$ has
  parts  $X_1$,  $\ldots$, $X_k$  that  are  all comparable  under  that
  partial order.  So there is an  ordering of $C$, $(P_1,  \ldots, P_k)$
  such    that     $X_1    \subset    \cdots    \subset     X_k$    (and
  $Y_k \subset  \cdots \subset Y_1$),  assuming that $P_i=  X_iY_i$, for
  $1 \leq i \leq k$.
  \qed
\end{proof}

Any other  biclique $Q = X_QY_Q$  that is mutually included  with $P_i$,
for some  $1 \leq i  \leq k$ is also  mutually included with  $P_j$, for
every $1 \leq j \leq i-1$ or every $i+1 \leq j \leq k$.

\begin{theorem}\label{teo.KBm-clique-order}
  For every graph $G$ and for every clique $C$ of $H = KB_m(G)$ there is
  an  ordering  of the  vertices  of  $C$  such  that for  every  vertex
  $Q \in V(H)$, with $N_{H}(Q) \cap C \neq \emptyset$, the vertices
  of $N_{H}(Q) \cap C$ are consecutive in that ordering and include the
  first or the last vertex of that ordering.
\end{theorem}
\begin{proof}
  Let $G$ be a graph and $C$ a  clique of $H = KB_m(G)$. That is, $C$ is
  a set  of bicliques of  $G$ such that every  two of them  are mutually
  included.    By   Lemma~\ref{lem.nested},   there   is   an   ordering
  $(P_1, \ldots, P_k)$ of the bicliques  of $C$.  Let $Q \in V(H)$, such
  that $N_{H}(Q) \cap C \neq \emptyset$.   Suppose that $Q = X_QY_Q$ and
  $P_i= X_iY_i$, for $1 \leq i \leq k$.

  Suppose that $X_Q \subset X_{\ell}$ for  some $1 \leq \ell \leq k$ and
  that   $\ell$    is   minimum.    Then
  $X_Q \subset X_i$   for every $\ell \leq i \leq k$.
  That    is,   $Q$    is    mutually   included    with   $P_i$,    for
  $\ell \leq i \leq k$.

  Now suppose that $X_r \subset X_Q$ for some $1 \leq r \leq k$ and that
  $r$ is  maximum.  By  the same argument,  $Q$ is
  mutually included with $P_i$, for $1 \leq i \leq r$.
  \qed
\end{proof}

Using Theorem~\ref{teo.KBm-clique-order} can be proved that some structures
do not occur in $KB_m(G)$. 

Let $T_{\net} = \{x_1, x_2, x_3\}$ be the vertices of the triangle of the \net\ graph 
(Figure~\ref{fig:netgraph}) and let $S_{\net} = \{s_1, s_2, s_3\}$ be the other vertices 
such that each $x_is_i$ is an edge.
Let $\net^*$ be any graph generated by a \net\ graph with 
0  or more  edges  added  connecting only  vertices  of $S_{\net}$.  See
Figure~\ref{fig:netstar}.

\begin{figure}[htb]
  \centering
  \includegraphics[width=3cm]{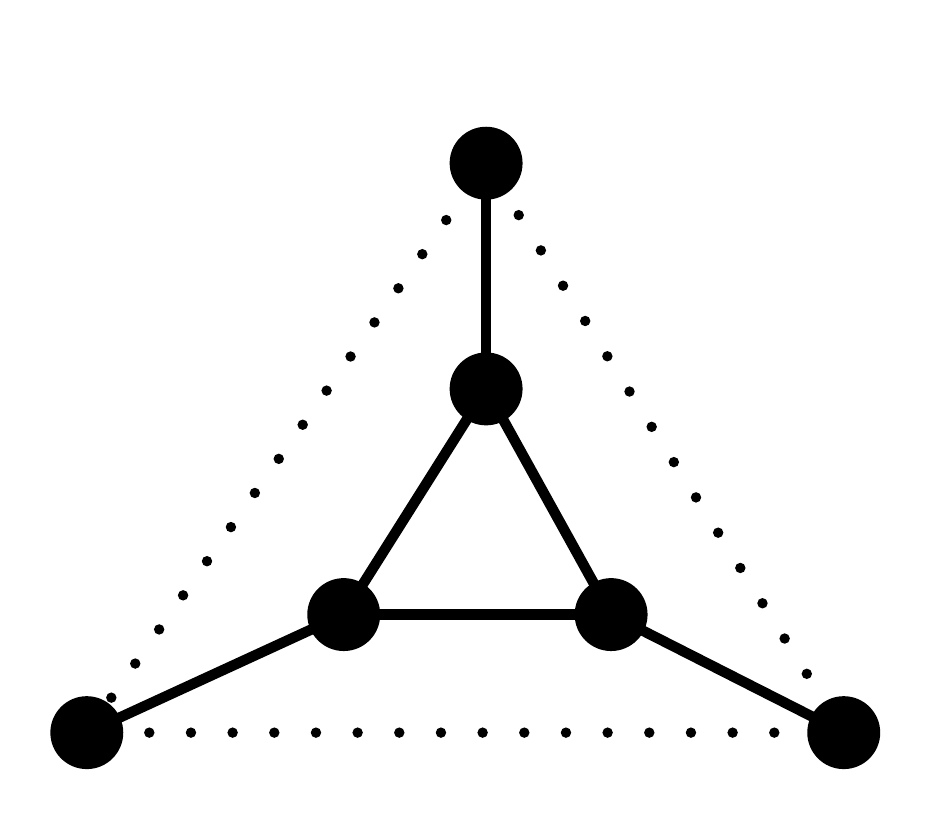}
  \caption{$\net^*$ structure. The dotted edges are optional.}
  \label{fig:netstar}
\end{figure}

\begin{corollary}\label{cor.netstarfree}
  For every graph $G$, $KB_m(G)$ does not contain a $\net^*$ as an induced subgraph.
\end{corollary}
\begin{proof}
  Let $G$ be a graph and $H = KB_m(G)$. Suppose there is a $\net^*$ graph 
  as an induced subgraph of $H$.
  As each vertex of $S_{\net}$ is adjacent to only one vertex of $T_{\net}$, 
  and $T_{\net}$ is a triangle,
  by   Theorem~\ref{teo.KBm-clique-order}  there   is  an   ordering  of
  $T_{\net}$.  Considering the neighbourhoods  of $s_1$ and $s_3$, $x_1$
  and $x_3$ are the first and the last (or the inverse) of that ordering. 
  Also  by Theorem~\ref{teo.KBm-clique-order},  the neighbourhood
  of $s_2$ should include $x_1$ or $x_3$, but $N_H(s_2) = x_2$.
  Therefore $H$ does not contain a $\net^*$ as an induced subgraph.
  \qed
\end{proof}

Now  we present  a  property about  square graphs  that  implies in  a
property of biclique graphs of triangle-free graphs.

Let $D_n$ be a graph with vertex set
$V(D_n) = \{v, u_1, \ldots, u_n, w_1, \ldots, w_n\}$ and edge set
$E(D_n) = \{ vu_i \mid 1 \leq i \leq n\} \cup \{ u_iw_i \mid
1 \leq i \leq n\}$.  Let $D^-_n$ be the graph $D_n - w_n$.
See Figures~\ref{fig:Dn} and~\ref{fig:Dn-}. 

\begin{figure}[htb]
\centering
\begin{subfigure}{.5\textwidth}
  \centering
  \includegraphics[width=3cm]{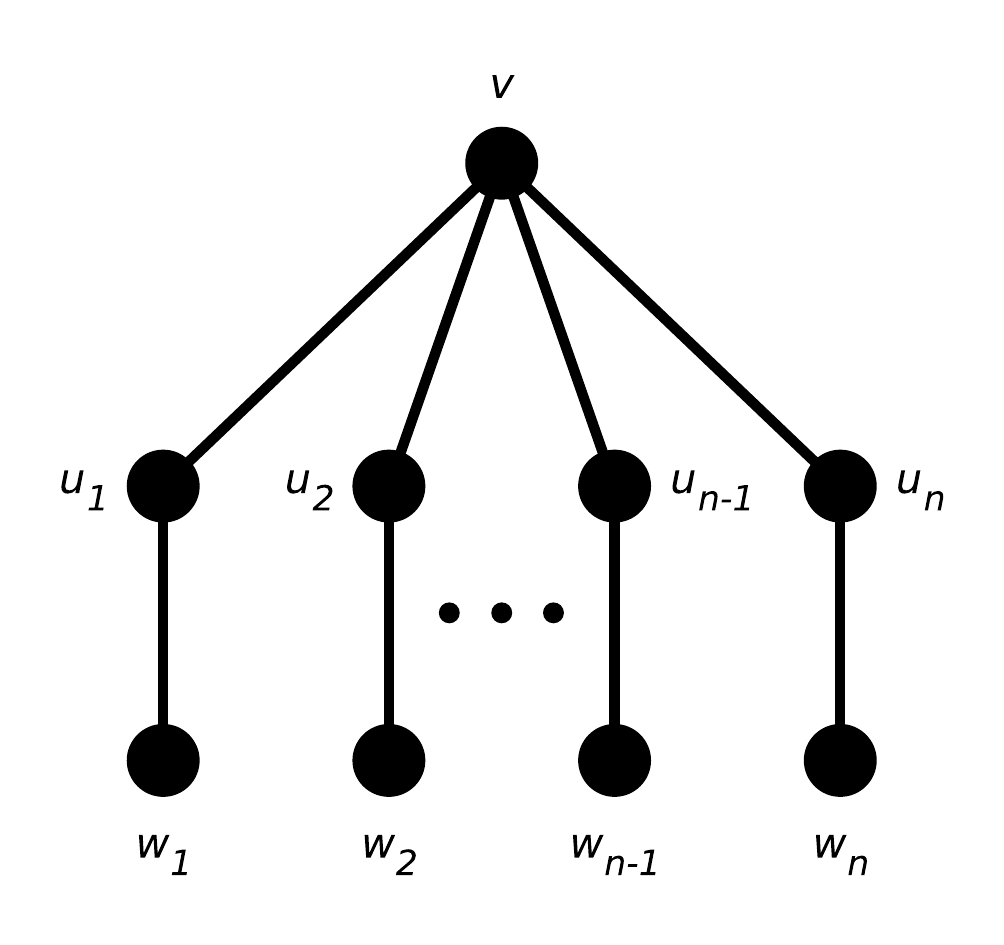}
  \caption{}
  \label{fig:Dn}
\end{subfigure}%
\begin{subfigure}{.5\textwidth}
  \centering
  \includegraphics[width=3cm]{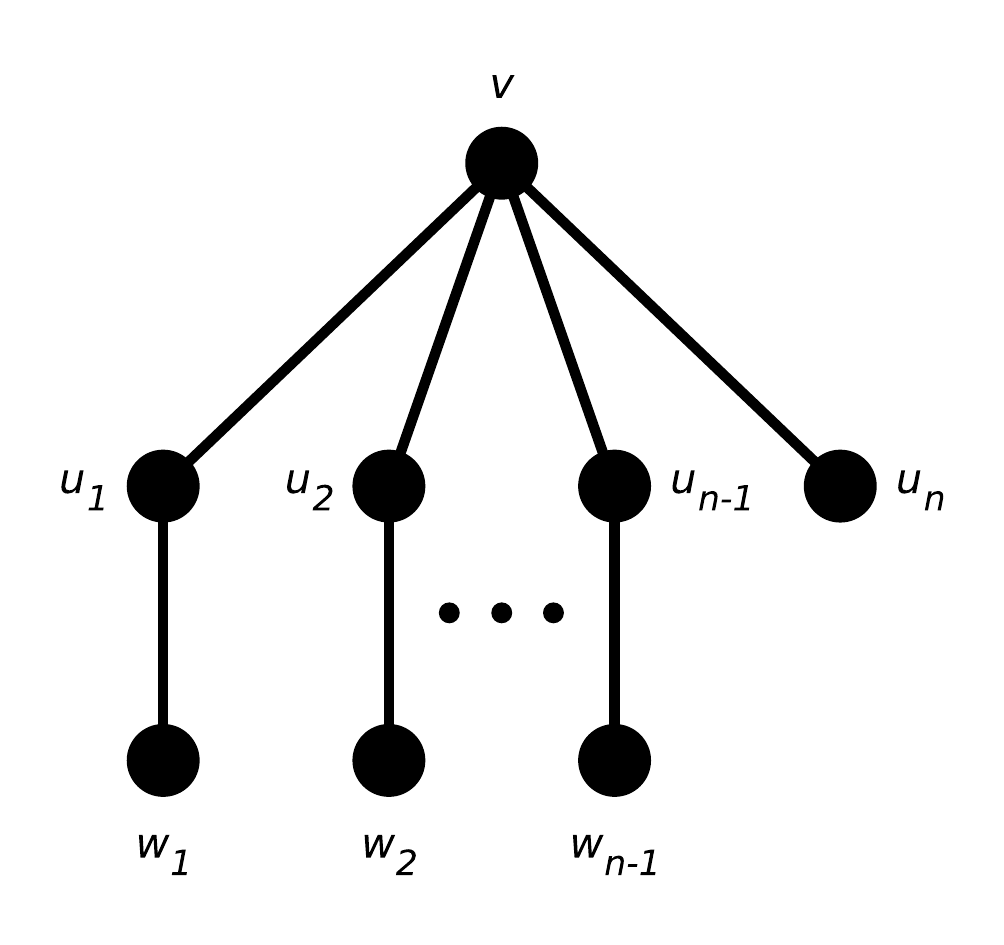}
  \caption{}
  \label{fig:Dn-}
\end{subfigure}%
\caption{(a) $D_n$ and (b) $D^-$ graphs.}
\end{figure}

Call ``old'' edges the edges of $G^2$ (for some graph $G$) that are
also edges of $G$ and call the other edges of $G^2$ as ``new'' edges.
That is, ``new'' edges are the edges that are in $G^2$ and do not
exist in  $G$.

\begin{theorem}\label{thm.Dnsquare}
  Let $H=G^2$ for some graph $G$. Then, every induced $K_{1,n}$ of $H$ is contained in a $(D_n)^2$ or a $(D^-_ n)^2$.
\end{theorem}
\begin{proof}%[sketch]
  Let  $H  = G^2$ for some graph $G$.   
  Suppose $S= \{v, w_1,\cdots, w_n\}$ induces a $K_{1,n}$ in $H$. 
  Observe first that $H[S]$ can contain at most one old
  edge. %otherwise its extremes $u_iu_j$ would be adjacent in $H$.
  If $e_i=vw_i$ is a new edge, then there  is a vertex  $v_{e_i}$ of
  $G$ that  is adjacent in  $G$ to $v$ and to $w_i$. Observe that $v_{e_i}$ can not be adjacent to any $w_j$, $j \neq i$ in $G$. Then, for every new edge $e_k=vw_k$ of $H[S]$ there is in $G$ a (different) vertex $v_{e_k}$ adjacent to $v$ and $w_k$. 
  
  %Therefore, i
  If all edges of $H[S]$  are new edges, vertices $v, v_{e_1}, \ldots, v_{e_n}$, $w_1, \ldots$, $w_n$ induce 
 a $D_n$ in $G$. Consequently $H[S]$ is contained in $(D_n)^2$.
 
 Otherwise, w.l.o.g. suppose $e_n$ is an old edge of $H[S]$. Then,  $G[\{v, v_{e_1}, \cdots$, $v_{e_{n-1}}$, $w_1, \cdots, w_n\}]$ is isomorphic to $(D^-_ n)$ and $H[S]$ 
 is included in $(D^-_ n)^2$.
  \qed
\end{proof}

See in Figures~\ref{fig:Dn2} and~\ref{fig:Dn-2} how a $K_{1,n}$ is included in a $(D_n)^2$ or a $(D^-_n)^2$.

\begin{figure}[htb]
\centering
\begin{subfigure}{.5\textwidth}
  \centering
  \includegraphics[width=4cm]{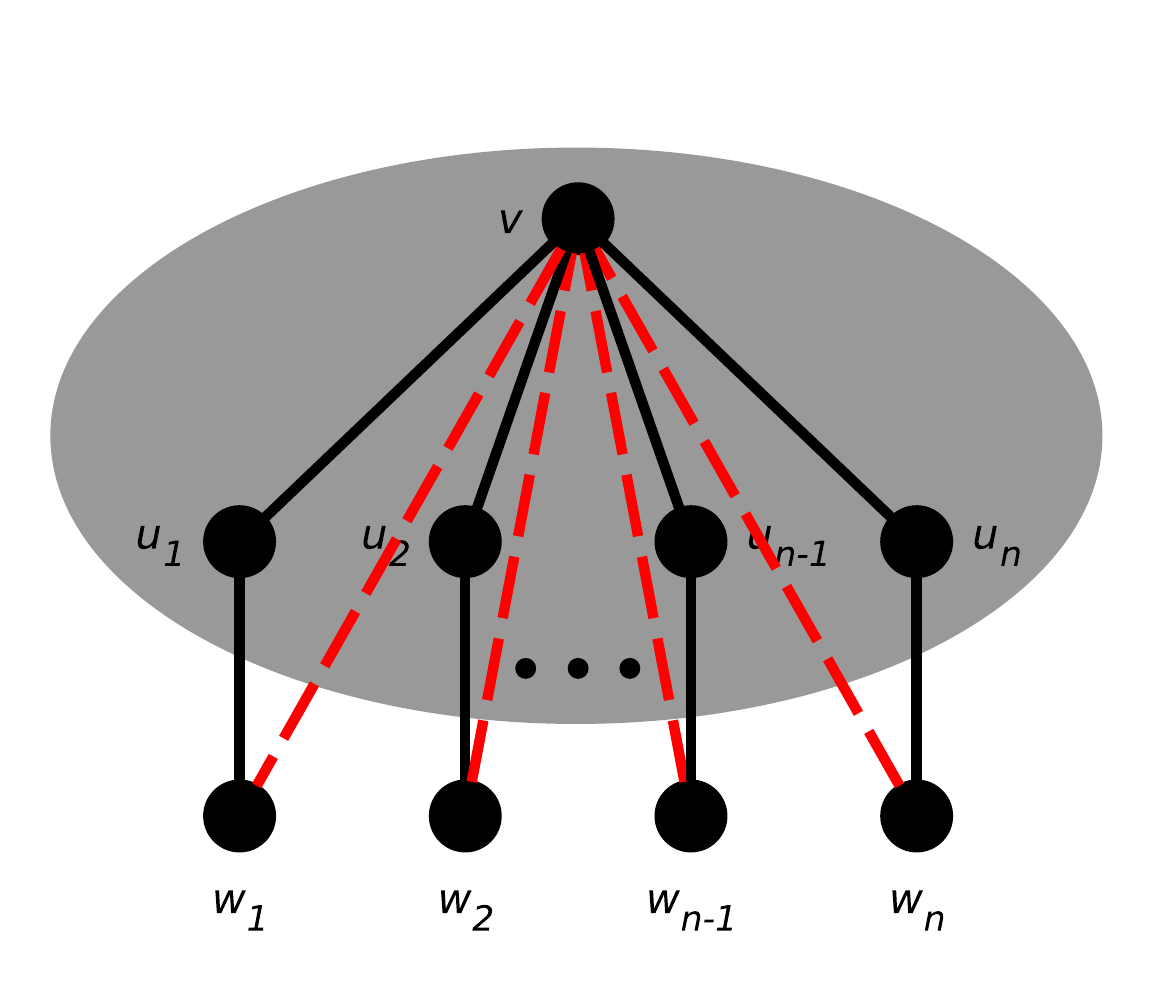}
  \caption{}
  \label{fig:Dn2}
\end{subfigure}%
\begin{subfigure}{.5\textwidth}
  \centering
  \includegraphics[width=4cm]{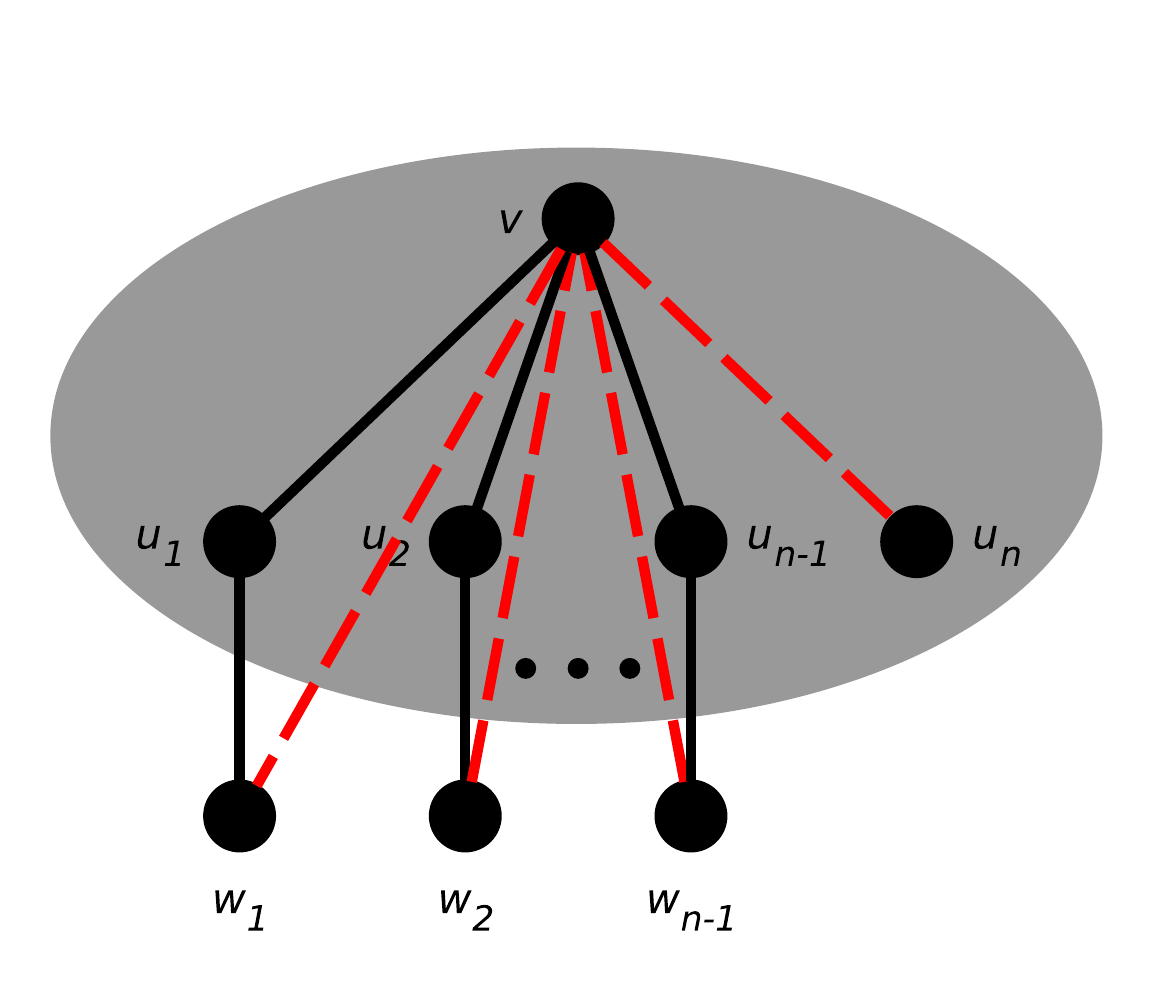}
  \caption{}
  \label{fig:Dn-2}
\end{subfigure}%
\caption{(a) $(D_n)^2$ and (b) $(D^-_n)^2$ graphs. The vertices in the grey area form a clique ($\{v,u_1, \ldots, u_n\}$). The edges of the $K_{1,n}$ are marked in dashed lines.}
\end{figure}

That property (of biclique graphs of triangle-free graphs) is a generalization 
of a very used property involving $P_3$ in
biclique graphs. Also, we think that this new property can be a useful tool
for solving the biclique graph recognition problem.
%This property also allows us to partially prove Conjecture~\ref{conj.montero}.

\begin{corollary}\label{cor.DnKB}
  Every induced $K_{1,n}$ of a biclique  graph of triangle-free graph is
  contained in a $(D_n)^2$ or a $(D^-_ n)^2$.
\end{corollary}

Note that a $P_3$ is a $K_{1,2}$, $(D_2)^2$ is a $3$-fan and $(D^-_2)^2$
is a diamond. Then, the fact that  every $P_3$ is contained in a diamond
or  a  $3$-fan  \cite{Groshaus2010}, when restricted to biclique graphs of triangle-free graphs,
is  a particular case of Corollary~\ref{cor.DnKB}.

% \begin{corollary}
%   Every induced  $P_3$ of a biclique  graph of a triangle-free graph is
%   contained in a diamond or a gem.
% \end{corollary}

We present the following result which generalizes the Conjecture~\ref{conj.montero} when restrict to biclique graph of triangle-free graphs.

\begin{theorem}\label{thm.independent}
  Let $G$ be a biclique graph of a triangle-free graph.  If there exists
  an  independent set  $I$ in  $G$  with $|I|  =  n \geq  3$, such  that
  $|\bigcap_{v      \in     I}      N_G(v)|      \geq     3$.       Then
  $|\bigcup_{v \in I} N_G(v)| > n$.
\end{theorem}
\begin{proof} 
  By    Theorem~\ref{thm.KBm-square},     $G=(KB_m(H))^2$    for    some
  triangle-free  graph  $H$.   Suppose   there  is  an  independent  set
  $I  =  \{v_1,  \ldots, v_n\}$  in  $G$  with  $n  \geq 3$,  such  that
  $|\bigcap_{v \in I} N_G(v)| \geq 3$.

  Consider  vertex $v_i$.   We  affirm  that there  exists  an old  edge
  $e_i=v_iw_i$, $w_i \in N_G(v_i)$. Let $w$ be a neighbour of $v_i$.  If
  $wv_i$ is not  an old edge, there exists $w_i  \in N_G(v_i)$ such that
  $v_iw_i$, $v_iw$ are  edges in $KB_m(H)$, that is  $v_iw_i$ and $v_iw$
  are old edges.  Observe that $w_i \neq w_j$ when $i \neq j$, since two
  vertices  of  $I$ does  not  have  a  common neighbour  in  $KB_m(H)$,
  otherwise    they   would    be    adjacent    in   $G$.     Therefore
  $|\bigcup_{v \in I} N_G(v)| \geq n$.

  Suppose  there are  only one  old edge  incident to  each vertex  of $I$,
  otherwise $|\bigcup_{v \in I} N_G(v)| > n$.
  Suppose                         w.l.o.g.                          that
  $W =  \{w_1, w_2, w_3\} \subseteq  \bigcap_{v \in I} N_G(v)$.   That is,
  $v_1w_1$, $v_2w_2$ and $v_3w_3$ are  old edges.  Since $v_1w_2$ exists
  and is  a new edge  then, there exists  a vertex $w_{1,2}  \in N_G(v_1)$
  such  that $v_1w_{1,2}$  and  $w_1w_{1,2}$ are  old  edges. Note  that
  $w_{1,2} = w_1$  as $v_1w_1$ is the unique old  edge incident to $v_1$
  and we conclude that $w_1w_2$ is an old edge.  Now following the same
  arguments for 
  %and considering  
  edges $v_2w_3$ and $v_3w_1$  we obtain that
  $w_1w_2$,  $w_2w_3$ and  $w_1w_3$  are old  edges.  Finally,  vertices
  $v_1,v_2,v_3,w_1,w_2,w_3$ induce a $\net^*$ in $KB_m(H)$ what leads to
  a contradiction according to Corollary~\ref{cor.netstarfree}.
  Consequently, $|\bigcup_{v \in I} N_G(v)| > n$.
  \qed
\end{proof}

Observe that in Theorem \ref{thm.independent} we do not ask for vertices to have the same neighbourhood. 
Finally, we prove the Conjecture~\ref{conj.montero}  for the case of biclique
graphs    of    triangle-free    graphs    as    a   corollary    of
Theorem~\ref{thm.independent}.

\begin{corollary}\label{cor.conjectureGM}
  Let $G$ be a graph not isomorphic to   the    diamond. 
  Suppose $v_1,     \ldots$,     $v_n$   
  are vertices of $G$ with $n \geq 2$,  such    that
  $N_G(v_1) = N_G(v_2) = \cdots =  N_G(v_n)$, with $|N_G(v_1)|=m$, for $m \leq n$. Then $G$ is not a biclique graph of a
  triangle-free graph $H$.
\end{corollary}
\begin{proof} 
  By contradiction, suppose  $G$ is the biclique graph  of some triangle-free graph.
  For  $m  =  2$,  if $N_G(v_1)=  \overline{K_2}$,  a  contradiction  is
  obtained  by  Corollary~\ref{cor.DnKB}.  The  case  $N_G(v_1)=K_2$  is
  proved     by      Groshaus     and      Montero     \cite[Proposition
  4.6]{GroshausMontero2019}.
  
  The case $m \geq 3$ follows directly by Theorem~\ref{thm.independent}
%  
  %So, assume  $n \geq 3$. As  $N_G(v_1) = N_G(v_2) =  \cdots = N_G(v_n)$
  %the  set  $I  =  \{v_1,  \ldots, v_n\}$  is  an  independent  set  and
  %$|\bigcap_{v \in I} N_G(v)|  \geq 3$. By Theorem~\ref{thm.independent}
  %$|\bigcup_{v \in I}  N_G(v)| > n$. What leads to  a contradiction since
  %$|N_G(v_1)| \leq n$. 
  \qed
\end{proof}

\section{Mutually Included Biclique Graphs of Bipartite Graphs}

In the  case of  bipartite graphs  the parts of  the bicliques  are also
bipartitioned, and we can establish a partial order of the bicliques based
on one part of the bipartite graph.

Let $G = (A  \cup B, E)$ be a bipartite graph.  Let  $\mB(G)$ be the set
of bicliques of  $G$.  Define the relation $\prec_G$  over $\mB(G)$ such
that $P \prec_G Q$ when $(P \cap  A) \subset (Q \cap A)$.  The reflexive
closure of $\prec_G$ is the partial order $\preceq_G$.

\begin{lemma}\label{lem.posetIIC}
  For  every   bipartite  graph  $G$   the  poset
  $(\mB(G), \preceq_G)$ is IIC.
\end{lemma}
\begin{proof}
% \comment{REV1}{green}{Proof  of Lemma  5,  second paragraph:  why does  it
%   happen that $I^-_{\cal P}(P)$ and  $I^-_{\cal P}(Q)$ do not intersect?
%   \\
% - If $P \cap Q \cap A = \emptyset$ no other biclique $R$ is such that $R
% \cap A \subset (P \cap Q \cap A)$.}
  Let  $G  = (A  \cup  B,  E)$ be  a  bipartite  graph, with  the  poset
  $\mP = (\mB(G), \preceq_G)$ and let $P$ and $Q$ be two different bicliques
  of $G$.

  If             $P             \preceq_G            Q$             then
  $P \in  I_{\mP}^-(P) \cap I_{\mP}^-(Q) =  I_{\mP}^-(P) \neq \emptyset$
  and  $P$ is  the maximum  of $I_{\mP}^-(P)  \cap I_{\mP}^-(Q)$.   Also
  $Q \in  I_{\mP}^+(P) \cap I_{\mP}^+(Q) =  I_{\mP}^+(Q) \neq \emptyset$
  and $Q$ is the minimum of $I_{\mP}^+(P) \cap I_{\mP}^+(Q)$.

  Now suppose $P$ and $Q$ are not comparable.

  If     $P     \cap     Q      \cap     A     =     \emptyset$     then
  $I_{\mP}^-(P) \cap I_{\mP}^-(Q) = \emptyset$.

  Suppose  $P  \cap Q  \cap  B  = \emptyset$  and  there  is a  biclique
  $R    \in   I_{\mP}^+(P)    \cap   I_{\mP}^+(Q)$.     By   definition,
  $(P      \cap     A)      \subseteq      (R      \cap     A)$      and
  $(Q    \cap     A)    \subseteq     (R    \cap    A)$,     that    is,
  $((P   \cup   Q)    \cap   A)   \subseteq   (R    \cap   A)$.    Then,
  $N_G^*(R   \cap    A)   \subseteq   (P    \cap   Q   \cap    B)$   and
  $(P     \cap    Q     \cap    B)     \neq    \emptyset$.     So,    if
  $P      \cap      Q      \cap       B      =      \emptyset$      then
  $I_{\mP}^+(P) \cap I_{\mP}^+(Q) = \emptyset$.

  If    $P    \cap    Q    \cap   A    \neq    \emptyset$    then,    by
  Lemma~\ref{lem.bicliqueintersecting},  there is  a  biclique $R$  such
  that $R \in I_{\mP}^-(P) \cap  I_{\mP}^-(Q) \neq \emptyset$ and $R$ is
  the maximum of $I_{\mP}^-(P) \cap I_{\mP}^-(Q)$. The same for part $B$
  and $I_{\mP}^+(P) \cap I_{\mP}^+(Q)$.

  So, the poset $\mP = (\mB(G), \preceq_G)$ is IIC.
  \qed
\end{proof}

We show that $KB_m($bipartite$)$ is an IIC-comparability graph.

\begin{lemma}\label{lem.KBm-comparability}
  If $G$  is a  bipartite graph then, $KB_m(G)$ is  an IIC-comparability
  graph.
\end{lemma}
\begin{proof}
  Let  $G  =  (A \cup  B,  E)$  be  a  bipartite graph  with  the  poset
  $(\mB(G),   \preceq_G)$.    By  Lemma~\ref{lem.mutually-subset},   two
  different bicliques, $P$ and $Q$, of $G$ are mutually included iff $P$
  and $Q$ are comparable by $\preceq_G$.

  So,   $KB_m(G)$   is   the    comparability   graph   of   the   poset
  $(\mB(G), \preceq_G)$.

  Moreover,     as     $(\mB(G),      \preceq_G)$     is     IIC,     by
  Lemma~\ref{lem.posetIIC}, $KB_m(G)$  is an IIC-com\-pa\-ra\-bi\-li\-ty
  graph.
  \qed
\end{proof}

Let     $\mP    =     (V,\leq)$    be     a    poset.      Define    the
\define{predecessors-successors              bipartite              graph
  $G_{\mP} = (A \cup B, E)$} as follows:
$A = \{ a_v \mid v \in V\}$;
$B = \{ b_v \mid v \in V\}$;
$E = \{ a_ub_v \mid u \leq v$, for $u, v \in V \}$.

Let   $X_v   =   \{a_u   \in   A  \mid   u   \in   I_{\mP}^-(v)\}$   and
$Y_v = \{b_w \in B \mid w \in I_{\mP}^+(v)\}$.

\begin{lemma}\label{lem.XvYv=bicliques}
  Given a  poset $\mP =  (V,\leq)$ and its  predecessors-successors graph
  $G_{\mP}$.  For every $v \in V$, $X_vY_v$ is a biclique of $G_{\mP}$.
\end{lemma}
\begin{proof}
  Let $\mP =  (V,\leq)$ be a poset and  its predecessors-successors graph
  $G_{\mP}$.
  
  Let $v \in V$.  As for every $a_u \in X_v$,
  $u \leq v$  and for every $b_w \in  Y_v$, $v \leq w$, then  $u \leq w$
  and $a_ub_w \in E(G_{\mP})$.  So $X_vY_v$ induces a complete
  bipartite subgraph of $G_{\mP}$.

  Now suppose $X_vY_v$  is not maximal. Then w.l.o.g. there  is a vertex
  $a_x  \not\in  X_vY_v$  such  that $X_vY_v  \cup  \{a_x\}$  induces  a
  complete bipartite subgraph  of $G_{\mP}$. By definition,  $x \leq w$,
  for every  $w \in  I_{\mP}^+(v)$. But $v  \in I_{\mP}^+(v)$,  and then
  $x \leq v$. Consequently, $a_x \in  X_v$ and $X_vY_v$ is a biclique of
  $G_{\mP}$.
  \qed
\end{proof}

For    every    subset    $S     \subseteq    A    \cup    B$,    define
$V(S) = \{v  \in V \mid a_v \in  S $ or $b_v \in S\}$  to be \define{the
  base set of $S$}.

% TODAS AS BICLIQUES DE $G$ SÃO DA FORMA $X_vY_v$?
\begin{lemma}\label{lem.XY=XvYv}
  Given an  IIC poset  $\mP =  (V,\leq)$ and  its predecessors-successors
  graph $G_{\mP}$.  Every biclique of $G_{\mP}$ is equal to $X_vY_v$ for
  some $v \in V$.
\end{lemma}
\begin{proof}
  Let $\mP  = (V,\leq)$ be  an IIC poset and  its predecessors-successors
  graph $G_{\mP}$. And let $XY$ be any biclique $XY$ of $G_{\mP}$.

  Consider the  sets $V(X)$ and  $V(Y)$.  Suppose  $V(X)$ do not  have a
  maximum  element. Then  there are  at  least two  maximal elements  in
  $V(S)$,     $i$    and     $j$.      As    $\mP$     is    IIC     and
  $I_{\mP}^+(i) \cap  I_{\mP}^+(j) \neq  \emptyset$ (as $a_i$  and $a_j$
  have    at    least    one    neighbour    in    common    in    $Y$),
  $I_{\mP}^+(i)   \cap   I_{\mP}^+(j)$   has   a   minimum,   $m$.    As
  $V(Y)  \subseteq   I_{\mP}^+(i)  \cap  I_{\mP}^+(j)$,  $a_m   \in  X$,
  consequently $m \in V(X)$. Moreover, $i  \leq m$ and $j \leq m$, which
  implies that  $i$ and  $j$ are  not maximal  elements of  $V(X)$.  So,
  $V(X)$ has  a maximum.   Using a  similar reasoning  we can  show that
  $V(Y)$ has a minimum. Let $v$  be the maximum element of $V(X)$.  Then
  $v  \in V(Y)$  and $v$  is  the minimum  element of  $V(Y)$. That  is,
  $XY = X_vY_v$.
  \qed
\end{proof}

We show now that the class  of IIC-comparability is exactly the class of
$KB_m($bipartite$)$. 

\begin{theorem}\label{thm.KBmBipartite}
  $KB_m($bipartite$) = $ IIC-comparability.
\end{theorem}
\begin{proof}
  By Lemma~\ref{lem.KBm-comparability}, $KB_m($bipartite$) \subseteq $
  IIC-comparability.

  Now let  $H$ be an IIC-comparability  graph, $H$, with its  IIC poset,
  $\mP  =   (V(H)$,  $\leq)$,  and  its   predecessors-successors  graph
  $G_{\mP}$. Recall that $G_{\mP}$ is a bipartite graph.

  By Lemmas~\ref{lem.XvYv=bicliques}  and~\ref{lem.XY=XvYv}, there  is a
  bijection  $\phi$  between the  vertex  set  of  $H$  and the  set  of
  bicliques of $G_{\mP}$, given by $\phi(v) = X_vY_v$.

  Let $u, v \in V(H)$, with $u \neq v$.

  Suppose $uv \in E(H)$. Then $u \leq v$ or $v \leq u$ (as $H$ is a
  comparability  graph with  poset $\mP$).   So $X_u  \subseteq X_v$  or
  $X_v \subseteq  X_u$ and $X_uY_u$  and $X_vY_v$ are  mutually included
  (by         Lemma~\ref{lem.mutually-subset})          and         then
  $\phi(u)\phi(v) \in E(KB_m(G_{\mP}))$.

  Now   suppose    $\phi(u)\phi(v)   \in    E(KB_m(G_{\mP}))$.   Then
  $X_u  \subseteq X_v$  or $X_v  \subseteq X_u$.  As $a_u  \in X_u$  and
  $a_v \in X_v$, then $u \leq v$ or $v \leq u$ and $uv \in E(H)$.

%\comment{REV1}{green}{Last paragraph  of the  proof of Theorem  3: this
%paragraph has many “and”’s, you should try to rewrite it;}
  So,     $\phi$     is     an     isomorphism     and     $H     \simeq
  KB_m(G_{\mP})$. Consequently IIC-com\-pa\-ra\-bi\-li\-ty $\subseteq KB_m($bipartite$)$.

  Therefore, $KB_m($bipartite$) = $ IIC-comparability.
  \qed
\end{proof}

Then,              considering             Theorems~\ref{thm.KBm-square}
and~\ref{thm.KBmBipartite} we conclude  the characterization of biclique
graphs of bipartite graphs.

\begin{corollary}\label{cor.Snstar}
  $KB($bipartite$) = ($IIC-comparability$)^2$.
\end{corollary}
\begin{proof}
  By Theorems~\ref{thm.KBm-square} and~\ref{thm.KBmBipartite}.
  \qed
\end{proof}

\section{Final Remarks}

In this work we prove that the biclique graph of a $K_3$-free graph is the square of
some particular graph called $KB_m$ extending all known properties of square graphs to biclique graphs of triangle-free graphs and providing a tool to prove other properties. Some published properties about biclique graphs can be easily proved if we restrict to square graphs.

We prove that      $KB(\mG)     \not\subseteq      (\mG)^2$      and $(\mG)^2  \not\subseteq KB(\mG)$ by presenting the graphs $4$-\wheel\ (Figure~\ref{fig:4wheel}),  that is
the biclique graph of some graph but  is not the square of any graph, and the square graph of  the \net\  (Figure~\ref{fig:netgraph}),  that is  not a  biclique graph  of any  graph. 
 We conclude that the  biclique  graphs of  $K_3$-free
graphs are  in the intersection  $KB(\mG) \cap  (\mG)^2$, but it  is not
known   if  $KB(K_3$-free$)   \subsetneq   KB(\mG)   \cap  (\mG)^2$   or
$KB(K_3$-free$) = KB(\mG) \cap (\mG)^2$. 

%\ag{Tiramos essa parte da cintura?}
%Knowing  the girth of  the $KB_m$ graphs  (square roots of  the biclique
%graphs of  $K_3$-free graphs) can  help to understand the  complexity of
%the  recognition  problem, since  recognizing  square  graphs is  solved
%depending on  the girth.   
\sloppy{Moreover, we give the first known property (which was conjectured by Groshaus and Montero) of biclique graphs that does not hold for square graphs. }

We study properties of $KB_m$ graphs  in terms of ordering  of the vertices of  their cliques and
neighbourhood  (Theorem~\ref{teo.KBm-clique-order}). That  property lead
to a forbidden structure ($\net^*$) for mutually included biclique graphs.
We need to find other
properties  of these  graphs in  order to  better understand  their square
graphs.

%About the square graphs, we present  a property that generalize the fact
%that every  $P_3$ is contained  in a diamond  or a $3$-fan  for biclique
%graphs of a $K_3$-free graphs.

%Using these  two latter properties  we partially proved a  conjecture by
%Groshaus  and   Montero  about   a  forbidden  structure   for  biclique
%graphs. Recall  that the class of  biclique graphs is not  an hereditary
%class. \ag{???}

We also show that the class of  biclique graphs of
bipartite graphs are exactly the square of a subclass of comparability
graphs (IIC-comparability). These characterizations (partial in the case
of  $K_3$-free  graphs)  do  not lead  to  polynomial  time  recognition
algorithms. However, it  gives another  tool to  study the  problem of
recognizing biclique graphs and their properties.

\bibliographystyle{splncs04}
\bibliography{bicliques}

% referências para posets:

% \url{https://www.whitman.edu/mathematics/higher_math_online/section05.05.html}\\
% Theorem 5.5.4 Any  partially ordered set is isomorphic to  a subset of a
% power set, ordered by the subset relation.

% \url{http://people.math.gatech.edu/~belinfan/6580su08/pdf/posets.pdf} \\
% Corollary. Any poset  $(P,\leq)$ is isomorphic to a  poset whose members
% are sets, partially ordered by inclusion.

% 
\end{document}